\documentclass[11pt,xcolor=dvipsnames]{amsart}

\usepackage[colorlinks=true, pdfstartview=FitV, linkcolor=blue, 
citecolor=blue]{hyperref}

\usepackage[dvipsnames]{xcolor}
\usepackage{amssymb,amsmath,amscd}
\usepackage{amsthm}

\usepackage{graphicx}
\usepackage{tikz}
\usepackage{bbm}
\usepackage{a4wide}
\usepackage{enumitem}
\usepackage{verbatim}
\usepackage{float}
\usepackage{accents}
\usepackage{stackrel}
\usepackage{faktor}
\usepackage[dvipsnames]{epsfig}
\usepackage[dvipsnames]{psfrag}  
\usepackage{pstricks}
\usepackage{pstricks-add}
\usepackage{etex} 
\usepackage{subfig}

\usepackage{pgfplots}
\usepackage{eufrak}
\usepackage{array}
\usepackage{cancel}
\usepackage{enumitem}

\newcolumntype{C}[1]{>{\centering\arraybackslash}p{#1}}

\pgfplotsset{width=6cm, compat = 1.15}
\usepackage[compress,square,comma,authoryear]{natbib}
\setcitestyle{citesep={;},open={(},close={)},aysep={}}

\theoremstyle{plain}
\newtheorem{theorem}{Theorem}[section]
\newtheorem{prop}[theorem]{Proposition}
\newtheorem{lemma}[theorem]{Lemma}

\theoremstyle{definition}
\newtheorem{remark}[theorem]{Remark}

\newtheorem{definition}[theorem]{Definition}

\renewcommand{\geq}{\geqslant}

\renewcommand{\leq}{\leqslant}

\newcommand{\ts}{\hspace{0.5pt}}

\newcommand{\RR}{\mathbb{R}\ts}
\newcommand{\PP}{\mathbb{P}\ts}

\newcommand{\NN}{\mathbb{N}}

\newcommand{\EE}{\mathbb{E}}

\newcommand{\cB}{\mathcal{B}}
\newcommand{\cC}{\mathcal{C}}

\newcommand{\cE}{\mathcal{E}}

\newcommand{\cG}{\mathcal{G}}
\newcommand{\cO}{\mathcal{O}}
\newcommand{\cP}{\mathcal{P}}

\newcommand{\cR}{\mathcal{R}}

\newcommand{\cT}{\mathcal{T}}

\newcommand{\one}{\mathbbm{1}}

\newcommand{\Psireco}{\Psi_{\text{rec}}}
\newcommand{\Psisel}{\Psi_{\text{sel}}}

\newcommand{\Tsel}[2]{{T^{#2}_{\textnormal{sel},#1}}}
\newcommand{\rhosep}{{\varrho_{\textnormal{sep}}}}
\newcommand{\nosep}{{\textnormal{ns}}}

\newcommand{\rhonosep}{{\varrho_\textnormal{ns}}}

\newcommand{\owl}[1]{\,\overline{\! #1 \!}\,}

\newcommand{\ee}{\mathrm{e}}

\newcommand{\dd}{\, \mathrm{d}}

\definecolor{gre}{rgb}{.06,.49,0.03} 

\DeclareMathOperator{\exponential}{Exp}

\DeclareMathOperator{\id}{id}

\DeclareMathOperator*{\argmax}{arg\,max}

\newcommand{\defeq}{\mathrel{\mathop:}=}
\newcommand{\eqdef}{=\mathrel{\mathop:}}

\definecolor{green2}{rgb}{0.0,0.8125,0.0}
\makeatletter
\newcommand*{\bigs}[1]{{\hbox{$\left#1\vbox to36\p@{}\right.\n@space$}}}
\makeatother

\begin{document}

\title[Selection, recombination, and the ancestral initiation graph] {Selection, recombination, and the ancestral initiation graph}

\author{Frederic Alberti, Carolin Herrmann, and Ellen Baake}

\address{\{Faculty of Mathematics, Faculty of Technology\}, Bielefeld University, 
\hspace*{\parindent}Postbox 100131, 33501 Bielefeld, Germany}
\email{\{falberti,ebaake\}@math.uni-bielefeld.de}

\address{Institute of Biometry and Clinical Epidemiology, Charit\'{e} - Universit\"atsmedizin Berlin,
\hspace*{\parindent} Charit\'{e}platz 1, 10117 Berlin, Germany}
\email{carolin.herrmann@charite.de}

\begin{abstract}
Recently, the selection-recombination equation with a single selected site and an arbitrary number of neutral sites was solved by \citet{Selrek} by means of the ancestral selection-recombination graph. Here, we introduce a more accessible approach, namely the ancestral initiation graph. The construction is based on a discretisation of the selection-recombination equation. We apply our method to systematically explain a long-standing observation concerning the dynamics of  linkage disequilibrium between two neutral loci hitchhiking along with a selected one. In particular, this clarifies the nontrivial dependence on the position of the selected site.
\end{abstract}

\maketitle

\emph{keywords:} selection-recombination differential
equation; ancestral initiation graph; linkage disequilibrium; hitchhiking; population genetics.

\section{Introduction}

The \emph{recombination equation} is a large nonlinear dynamical system that describes the evolution of the distribution of genetic types within an infinite population under the influence of \emph{recombination}, that is, the reshuffling of genetic information that occurs in the process of meiosis during the formation of germ cells (or gametes) in sexually reproducing populations.
Since its introduction by \citet{jennings}, \citet{robbins}, and \citet{geiringer}, the recombination equation has posed a major challenge to mathematical population geneticists. It was finally solved by~\citet{haldane} by considering a backward-time partitioning process that describes the random ancestry of a single individual. 

The logical next step was to attack the \emph{selection}-recombination equation, which describes the evolution under the additional influence of natural selection. This was previously considered unsolvable; in fact, the monograph by \citet{akin} starts with the words `The differential equations which model the action of selection and recombination are nonlinear equations which are impossible to solve explicitly'. While we do not  challenge this statement in its generality, \citet{Selrek} \emph{did} derive an explicit solution in the special case of a single selected site linked to a number of neutral sites, with single crossovers between the sites; this is particularly relevant in the context of   \emph{hitchhiking} \citep{MaynardSmithHaigh}, that is, the increase in frequency of neutral alleles linked to a beneficial mutation at the selected site. 
An approximate version of this selection-recombination equation was solved by \citet{stephansonglangley} for the case of two neutral loci linked to the selected one, with two alleles at each of the three loci. The  solution, which involves the incomplete Beta function, displays an interesting behaviour, which depends on whether the selected locus lies outside or between the neutral ones. While the approximation seems to be well justified in the parameter regime considered, where selection is much stronger than recombination, the resulting solutions are not easy to interpret, let alone generalise.

In contrast, \citet{Selrek} have recently obtained an \emph{exact} recursive solution of the full nonlinear system, again for a single selected locus and single-crossover recombination, but an \emph{arbitrary} number of neutral loci and an arbitrary position of the selected locus within the sequence. This solution involves intricate probabilistic constructions, based on the \emph{ancestral selection-recombination graph} by \citet{DonnellyKurtz} (see also \citet{LessardKermany}), as well as a generalisation of the notion of product measure. On a more abstract level, the authors proved  formal dualities between the solution of the selection-recombination equation and various stochastic processes with clear genealogical meaning.

The purpose of the present article is to complement this work in a number of ways. First, we will assume for the proofs (without loss of generality) that the selected locus is the first locus in the sequence, which eases  geometric intuition; we will indicate how this generalises to an arbitrary position of the selected site via an appropriate relabelling of the sites. Secondly, we introduce a novel \emph{ancestral initiation graph}, which arises naturally via discretisation of the selection-recombination equation and simplifies the genealogical arguments based on the ancestral selection-recombination graph. 

We apply our methods and results to the evolution of linkage disequilibrium between two linked neutral loci in the context of genetic hitchhiking. Similar to \citet{stephansonglangley} and \citet{PfaffelhuberLehnertStephan}, we consider two different geometries, with the selected locus located either in between or outside the two neutral loci. By a suitable reparametrisation, we give a unified treatment of both situations. While \citet{stephansonglangley} provide a purely numerical 
illustration of the time course, and   \citet{PfaffelhuberLehnertStephan} consider the more static picture with a focus on the structure of linkage disequilibrium at fixed times close to the time of fixation of the beneficial allele, we arrive at a thorough understanding, as well as a genealogical interpretation, of  the full  dynamics over time.

This paper is organised as follows. First, we recall the selection-recombination equation, along with the surrounding concepts  (Section~\ref{sec:SRE}).
Then, in Section~\ref{sec:stochastic}, we introduce the ancestral initiation graph, both in discrete and continuous time, and relate it to the constructions introduced by \citet{Selrek}. Subsequently, we use the ancestral initiation graph to give a probabilistic proof of the recursive solution; this is complemented by a more algebraic proof, which stays closer to the underlying discretisation scheme. In Section~\ref{sec:LDs}, we discuss the application to the dynamics of linkage disequilibrium; we close by discussing possible extensions and limitations of our approach in Section~\ref{sec:outlook}. 

\section{The selection-recombination equation and its solution}
\label{sec:SRE}
The \emph{selection-recombination equation} is a system of ordinary differential equations describing the evolution of the genotype distribution in an infinitely large, haploid population. Equivalently, one may consider a diploid population in the absence of dominance (that is, with fitness additive across gametes) and in Hardy--Weinberg equilibrium, and work at the level of gametes.  The finite set
$S \defeq \{1,\ldots,n\}$  
represents the genetic loci or \emph{sites} of interest.

We assume that there are two possible alleles at each site, denoted by $0$ and $1$. Thus, we think of \emph{$($genetic$)$ types} as binary sequences of length $n$, i.e. elements $x = (x^{}_1, \ldots, x^{}_n)$ of the \emph{type space}
$X = \{0,1\}^n$.
Since we will later also consider the evolution of (marginal) type distributions defined over subsets of loci, we define, for any nonempty $U \subseteq S$ and $x \in X$, the corresponding \emph{marginal type} \mbox{$x^{}_U \defeq (x^{}_i)^{}_{i \in U} \in  \{0,1\}^U \eqdef X_U$}, where $X_U$ is the \emph{marginal type space}.

We identify a population with its \emph{type distribution}, a probability measure (or vector) \mbox{$\nu = \big ( \nu(x) \big )_{x \in X} \in \cP(X)$}, where $\cP(X)$ denotes the set of all probability measures on $X$;  so, \mbox{$\nu(\{x\}) \geqslant 0$} denotes
the proportion of individuals of type $x$ in the population and
$\nu(E) := \sum_{x \in E} \nu(\{x\})$ for $E \subseteq X$.
We  often abbreviate $\nu(\{x\})$ as $\nu(x)$. Clearly, 
$\| \nu \| \defeq \| \nu \|_1 = \sum_{x \in X} \nu(x) =1$.

We define the \emph{marginal distribution} $\nu_U^{}$ of $\nu$ with respect to $U \subseteq S$ via
\begin{equation*}
\nu_U^{}(E) \defeq \nu (E \times X^{}_{S \setminus U}) \quad \text{for all } E \subset X_U.
\end{equation*}
In particular, for $x \in X_U$, $\nu^{}_U(x)=\nu(x, *)$, where we use `$*$'  as a shorthand for $X_{S \setminus U}$ (so $\nu(x, *)= \nu \big (\{x\} \times X_{S \setminus U}\big ) = \sum_{y \in X_{S \setminus U}} \nu(x, y)$). 

We describe the evolution of the type distribution by a time-dependent family $\omega =  ( \omega_t^{} )^{}_{t \geqslant 0}$ of distributions on $X$ (so that $\omega^{}_t(x)$ is the proportion of type $x$ at time $t$) that satisfies the \emph{selection-recombination equation} (SRE)
\begin{equation}\label{recoselequationprelim}
\dot{\omega^{}_t} = \Psi_{\text{sel}}(\omega^{}_t) + \Psi_{\text{reco}}(\omega^{}_t) \eqdef \Psi(\omega^{}_t);
\end{equation} 
here, the operators $\Psi_{\text{sel}}$ and $\Psi_{\text{reco}}$ describe the (independent) action of selection and recombination, which we now describe in more detail.
Regarding selection, we assume that the fitness of an individual\footnote{The reader should keep in mind that the following individual-based description merely serves to illustrate the model. We stress that we are working here with an infinite population in a law of large numbers regime. In particular, we neglect resampling. For the details of the connection between the finite-population Moran model and the SRE, see~\cite{Selrek}.} is determined by its allele at a single, fixed, site $i_\bullet \in S$, which we call the \emph{selected site}; an individual of type $x$ is \emph{fit} if $x_{i_\bullet}^{} = 0$ and \emph{unfit} if $x_{i_\bullet}^{} = 1$.  Fit individuals reproduce at rate $1+s$ with $s >0$, while unfit ones reproduce at rate 1, so $s$ is the selective advantage.  We write $f(\nu) \defeq \nu_{i_\bullet}^{}(0)$ for the proportion of fit individuals in a population $\nu$. In addition, we write $b(\nu)$ ($d(\nu)$) for the type distribution \emph{within the subpopulation of fit \textnormal{(}unfit\textnormal{)} individuals}. That is, $b(\nu)$ ($d(\nu)$) is the type distribution of an individual sampled from the population $\nu$, conditional on being fit (unfit).
More formally, $b(\nu)$ and $d(\nu)$ can be defined via
\begin{equation}\label{alternative}
f(\nu) b(\nu)(x) =  (1 - x^{}_{i_\bullet}) \nu(x) \eqdef F(\nu)(x)
\end{equation}
and
\begin{equation*}
\big ( 1 - f(\nu) \big ) d (\nu)(x)  \defeq x^{}_{i_\bullet} \nu(x) = \nu(x) - f(\nu) b(\nu) (x) = \nu(x) - F(\nu)(x),
\end{equation*}
respectively (note that the map $\nu \mapsto F(\nu)$ thus defined is linear). When an individual reproduces, its single offspring inherits the parent's type and replaces a randomly chosen individual in the population. The net effect of the aforementioned difference $s$ in the reproduction rate is that each individual in the population $\nu$ gets replaced, at total rate $sf(\nu)$, by a random fit individual. Thus, the selection part in Eq.~\eqref{recoselequation} reads
\begin{equation}\label{psisel}
\Psi_{\text{sel}}(\nu) = s f(\nu) \big (b(\nu)  - \nu \big ) = s \big (F(\nu) - f(\nu)\nu \big);
\end{equation}
note that  only  the difference between the reproduction rates enters $\Psi_{\text{sel}}$, because the baseline reproduction (which occurs at rate $1$) cancels out. 
\begin{remark}\label{negative selection}
In the formulation of the selection term, we assumed positive selection, i.e, $s > 0$. There would be no difficulty in allowing $s < 0$ in the sequel. However, this would not be a true generalisation, as it would merely switch the roles of $0$ and $1$ at the selected site; recall that only the difference in the reproduction rate matters. There is, however, a more subtle reason to stick with $s > 0$; it is well known and fundamental for the treatment in \cite{Selrek} that the solution of the selection equation is connected to a Yule process with branching rate $s$ (see also Remark~\ref{connectiontoessASRG}). Clearly, this only makes sense for $s > 0$. \hfill $\diamondsuit$
\end{remark}

Regarding recombination, we will restrict ourselves to single crossovers. For any $i \in S \setminus i_\bullet^{} \eqdef S^\circ$, we assume that, at rate
$\varrho_i^{} \geqslant 0$, new offspring are produced by two parents, so that the crossover occurs at site $i$. This means that the sequence of the offspring can be thought of as being fragmented, at the crossover site $i$, into two contiguous blocks, which we denote by $C_i$ and $D_i$ and call the ($i$-)head and ($i$-)tail, respectively. The tail starts at (and includes the) site $i$, while the head is defined as the complement of the tail and contains the selected site; it is called the \emph{head} because it contains all the information relevant for the fitness of the individual. More explicitly, we define
\begin{equation*}
(C_i,D_i) = 
\begin{cases}
\big ( [i+1,\ldots,n],[1,\ldots,i] \big ) \textnormal{ if } i < i_\bullet^{}, \\
\big ( [1,\ldots,i-1],[i,\ldots,n] \big ) \textnormal{ if } i > i_\bullet^{} ;
\end{cases}
\end{equation*}
the underlying mental picture is that recombination at site $i$ separates site $i$ from $i_\bullet^{}$, which leads us to exclude $i_\bullet^{}$ from the set of possible crossover sites. This way, we address the sites rather than the links between them (which would otherwise be more natural), which allows us to take the particular role of $i_\bullet^{}$ into account when formulating the recombination process.

We assume that each offspring individual inherits its alleles at the sites in $C_i$ from one parent, and those at the sites in $D_i$ from the other. Therefore, the offspring's type will be $x$ if the marginal types of its parents with respect to $C_i$ and $D_i$ are given by $x_{C_i}^{}$ and $x_{D_i}^{}$. Assuming random mating, the parents can be thought of as independent samples from the current population $\nu$, so that the offspring is of type $x$ with probability 
$\nu(\ast,x_{C_i}^{}) \nu(x_{D_i}^{},\ast)$, assuming $i < i_\bullet^{}$ without loss of generality (more precisely, the probability is $\nu( x_{C_i}^{},\ast) \nu(\ast,x_{D_i}^{})$ for $i > i_\bullet^{}$).
Put differently, the distribution of the offspring's type is given by
the product measure
\begin{equation}\label{defrecombinator}
\cR_i^{}(\nu) \defeq \nu_{C_i}^{} \otimes \nu_{D_i}^{}, \quad i \in S^\circ,
\end{equation}
where the operator $\cR_i^{} : \cP(X) \to \cP(X)$ thus defined is  called a \emph{recombinator} \citetext{\citealp{haldane,recoreview}; see also \citealp{canada}}.
The recombination part in Eq.~\eqref{recoselequationprelim} therefore reads
\begin{equation}\label{psireco}
\Psi_{\text{reco}}(\nu) = \sum_{i \in S^\circ} \varrho_i^{} \big ( \cR_i^{}(\nu) - \nu \big ).
\end{equation}
Putting \eqref{recoselequationprelim}, \eqref{psisel}, and \eqref{psireco} together, the \emph{selection-recombination equation} reads 
\begin{equation}\label{recoselequation}
\dot{\omega^{}}_t = s \big (F(\omega_t^{}) - f(\omega_t^{})\omega_t^{} \big) + \sum_{i \in S^\circ} \varrho_i^{} \big ( \cR_i^{}(\omega_t^{}) - \omega_t^{} \big ).
\end{equation}

\citet{Selrek} constructed the solution $\omega=(\omega^{}_t)^{}_{t \geqslant 0}$ of  \eqref{recoselequation} recursively via the family of solutions $\omega^{(k)}_{} = \big ( \omega_t^{(k)} \big)_{t \geqslant 0}^{}$ where $0 \leqslant k < n$, which interpolate between the solution of the pure selection-equation ($k=0$) and the solution of the SRE ($k=n-1$). For $i_\bullet^{} = 1$, $\omega^{(k)}$ solves the SRE \emph{truncated} at site $k$, 
\begin{equation}\label{recoselequationtruncated}
\dot \omega_t^{(k)} =  \Psi^{(k)} (\omega_t^{(k)}) \quad \text{with } \quad \Psi^{(k)} \defeq \Psisel + \Psireco^{(k)} \quad \text{and }
\Psireco^{(k)} \defeq \sum_{i=1}^{k} \varrho_{i+1}^{} \big ( \cR_{i+1}^{} - \id \!\big),
\end{equation}
with the same initial condition $\omega_0^{(k)} = \omega_0^{}$ for all $k$. We see that for $k = 0$, Eq.~\eqref{recoselequationtruncated} reduces to the pure selection equation $\dot \omega^{}_t = \Psi^{}_{\text{sel}}(\omega^{}_t)$, for $0<k<n-1$, sites $k+1, \ldots, n$ are `glued together', and for $k = n-1$, we recover Eq.~\eqref{recoselequation}. 

The following is~Theorem~5.4 of \citet{Selrek} in the special case $i_\bullet^{} = 1$.
\begin{theorem} \label{recursion}
For all $0 < k < n$, the solutions $\omega^{(k)}$ of Eq.~\eqref{recoselequationtruncated} satisfy the recursion
\begin{equation*}
\omega_t^{(k)} = \ee^{- \varrho_{k+1} t}\omega_t^{(k-1)} + \omega_{C_{k+1},t}^{(k-1)} \otimes \int_{0}^t \varrho_{k+1} \ee^{- \varrho_{k+1} \tau} \omega_{D_{k+1},\tau}^{(k-1)} \dd \tau,
\end{equation*}
starting with the solution
\begin{equation*}
\omega_t^{(0)} = \frac{\ee^{st} F(\omega_0^{}) + (\id - F)(\omega_0^{})}{\ee^{st} f(\omega_0) + (1 - f(\omega_0))}  \eqdef \varphi_t^{}(\omega_0^{})
\end{equation*} 
of the pure selection equation, whose flow we denote by $\varphi = (\varphi_t^{})_{t \geqslant 0}^{}.$  \qed
\end{theorem}

That $\omega_t^{(0)}$ indeed satisfies the pure selection equation can be verified  by a straightforward computation. Our goal is to prove the recursion for $0< k < n$, using a stochastic representation of $\omega$, which is related to Eq.~\eqref{recoselequation} via discretisation; in contrast, \citet{Selrek} relied on an underlying approximation by stochastic models for finite populations, and the corresponding duals.

\smallskip

\textbf{Arbitrary selected site.} Let us now generalise this recursion to an arbitrary choice of $i_\bullet^{} \in S$.
The idea is to relabel the indices in such a way that in any step of the iteration, the corresponding tail is not subdivided by any crossover event considered up to and including this step, but instead only separated from the selected site as an intact entity. We achieve this by considering relabellings that are nondecreasing with respect to the following partial order, meaning that they move outward from the selected site without leaving holes. 
\begin{definition}\label{porder}
For two sites $i,j \in S$, we say that $i$ precedes $j$, or $i \preccurlyeq j$, if either $i_\bullet \leq i \leq j$ or $i_\bullet \geq i \geq j$. We write $i \prec j$ if $i \preccurlyeq j$ and $i \neq j$. It is easy to check that the $i$-tail is, for $i \in S^\circ$ and independently of the position of $i$ with respect to $i_\bullet^{}$, given by
\begin{equation*}
D_i = \{j \in S  :  i \preccurlyeq j\}, 
\end{equation*}  
the set of sites that succeed $i$, including $i$ itself. Again, the $i$-head $C_i$ is the complement of the $i$-tail, $C_i \defeq S \setminus D_i = \owl{D_i}$ (throughout, the overbar will denote the complement with respect to $S$); see Figure \ref{headsntails}. Note that the  $i$-head always  contains the selected site. 
\end{definition}

\begin{figure}
\includegraphics[width = 0.85 \textwidth]{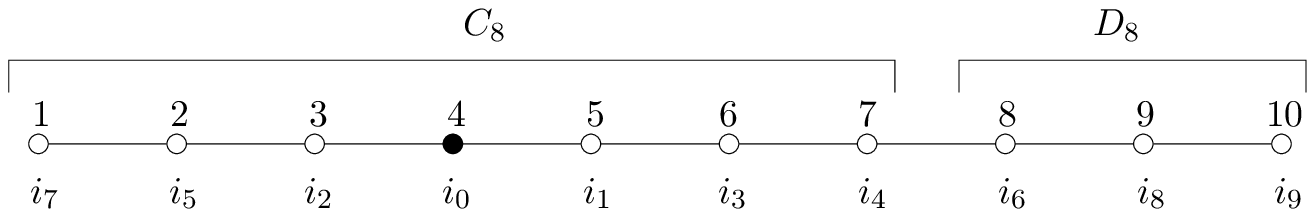}
\caption{\label{headsntails} A sequence of length 10 with selected site $i_\bullet=4$, an instance of head and tail, and relabelling of sites required for the recursive solution (the original site labels are at the top); see the text for more. 
}
\end{figure}


As announced above, we fix a nondecreasing (in the sense of the partial order from Definition \ref{porder}) relabelling  $(i_k)_{0 \leq k <n}$ of $S$ (cf. Fig.~\ref{headsntails}); note that this always forces $i_0 = i_\bullet^{}$, but otherwise it is, in general, not unique. For $0<k<n$, we then denote the corresponding heads, tails, and recombination rates by upper indices, that is,  $C^{(k)} \defeq  C_{i_k}$, $D^{(k)} \defeq  D_{i_k}$, $\cR^{(k)} = \cR_{i_k}$, and $\varrho^{(k)} \defeq \varrho^{}_{i_k}$. The SRE truncated at site $0 \leqslant k<n$ of \eqref{recoselequationtruncated} now turns into
\begin{equation}\label{recoselequationtruncated_gen}
\dot \omega_t^{(k)} =  \Psi^{(k)} (\omega_t^{(k)}) \quad \text{with } \quad \Psi^{(k)} \defeq \Psisel + \Psireco^{(k)} \quad \text{and }
\Psireco^{(k)} \defeq \sum_{\ell=1}^{k} \varrho^{(\ell)} \big ( \cR^{(\ell)} - \id \! \big),
\end{equation}
again with initial condition  $\omega_0^{(k)} = \omega_0^{}$  for all $k$. 
Finally, for an arbitrary position of $i_\bullet^{}$, the recursion of Theorem~\ref{recursion} reads \citep[Thm.~5.4]{Selrek}
\begin{equation}\label{rec_general}
\omega_t^{(k)} = \ee^{- \varrho^{(k)} t} \omega_t^{(k-1)} + \omega_{C^{(k)},t}^{(k-1)} \otimes  \int_{0}^t \varrho^{(k)} \ee^{- \varrho^{(k)} \tau} \omega_{D^{(k)},\tau}^{(k-1)} \dd \tau, \quad 0<k<n.
\end{equation} \hfill $\diamondsuit$

Upon a first read-through, the reader may want to restrict themself to the case $i_\bullet^{} = 1$, to ease geometric intuition. However, the following constructions do not depend on this choice.

\section{The ancestral initiation graph}
\label{sec:stochastic}
\label{sec:asrg}

The proof of \eqref{rec_general} (and Theorem~\ref{recursion} as a special case)
by \citet{Selrek} relied on a probabilistic interpretation of Eq.~\eqref{recoselequation} via a variant of the 
\emph{ancestral selection-recombination graph} (ASRG) \emph{in the deterministic limit}~\citep[Sec.~7]{Selrek}. The fundamental idea is to trace back the (random) ancestry of a fixed individual, by combining the ancestral selection graph (ASG) introduced by \citet{KundN} and the ancestral recombination graph (ARG) \citep{Hudson83,griffithsmarjoram96,griffithsmarjoram97}, adapted to the law of large numbers regime. However, as the arguments based upon this construction are somewhat delicate, we will instead proceed via
discretisation of the selection-recombination equation, explaining the genealogical content of our constructions as we go along.

We start by considering an Euler scheme for Eq.~\eqref{recoselequation}. Given a step length $h >0$, we set
\begin{equation}\label{definitionofT}
T_h^{}(\nu) \defeq \nu +  h \Psi(\nu)
\end{equation}
for all $\nu \in \cP(X)$. The operator $T_h^{}$ thus defined represents one step in the classical Euler method of numerical integration. The semigroup generated by it can be thought of as a discrete approximation to the flow of Eq.~\eqref{recoselequation}. Due to the Lipschitz continuity of this equation (cf.~\citet[Prop.~1]{haldane} for the Lipschitz continuity of the recombinators), standard results yield
\begin{equation}\label{eulerschemeoriginal}
\lim_{h \to 0} T^{\lfloor \frac{t}{h} \rfloor}_h (\omega_0^{}) = \omega_t^{}
\end{equation}
in norm, uniformly in $\omega_0^{}$ and locally uniformly in $t$; see \citet[Thm.~212A]{Butcher}, where \mbox{$T_h^n \defeq \underbrace{T_h^{} \circ \ldots \circ T_h^{}}_{n \; \text{times}}$} is the $n$-th power of $T_h^{}$.

Writing out the definition of $\Psi$,
\begin{equation*}
T_h^{}(\nu) = \Big ( \Tsel{h}{}(\nu) - h \! \sum_{i \in S^\circ} \varrho_i^{} \nu \Big ) + h \! \sum_{i \in S^\circ} \varrho_i^{} \cR_i^{} (\nu),
\end{equation*}
where
\begin{equation}\label{Tseldefinition}
\Tsel{h}{}(\nu) \defeq \nu + h \Psi_{\text{sel}}(\nu).
\end{equation}

Due to the nonlinearity of the recombinators and of $\Tsel{h}{}$ (via $f(\nu) \nu$), computing the powers of $T_h^{}$ as required in Eq.~\eqref{eulerschemeoriginal} is not an easy task. However, our life is made easier by the fact that $\Tsel{h}{}$
is linear on the level sets of $f$, that is,
\begin{equation} \label{linearity}
\Tsel{h}{}(\alpha \mu + (1 - \alpha) \mu') = \alpha \Tsel{h}{}(\mu) + (1- \alpha) \Tsel{h}{}(\mu')
\end{equation}
for all $\mu, \mu' \in \cP(X)$ with $f(\mu) = f(\mu')$ and $\alpha \in [0,1]$, which can be seen as follows. For each $\lambda \in [0,1]$, $\Tsel{h}{}$ is, on the level set $f^{-1}(\lambda)$, given by $(1 - hs\lambda) \id + hs F$. This is immediate from Eqs.~\eqref{psisel} and \eqref{Tseldefinition}; recall that $F$ is linear. 

 Biologically speaking, this is because the fitness of an individual depends only on the allele at the selected site; while $\Tsel{h}{}$ changes the relative sizes of the two subpopulations consisting of fit and unfit individuals respectively,  the type composition \emph{within} each subpopulation is not affected. Put differently, $\Tsel{h}{}$ acts in the same way on \emph{all} sequences in each subpopulation, increasing the weight of each fit sequence by a factor of $1+hs$, leaving the weight of each unfit sequence unchanged, and normalising by the total weight $1 + hsf(\nu)$ (this reflects that the the total population size is kept constant since offspring replace randomly chosen individuals, irrespective of their type). This changes the proportion of fit individuals 
\begin{equation*}
\text{from } \; f(\nu) \; \text{ to  } \; \frac{(1+hs)f(\nu)}{1 + hsf(\nu)}
\end{equation*}
and the proportion of unfit individuals 
\begin{equation*}
\text{from } \;  1 - f(\nu) \; \text { to } \; \frac{1 - f(\nu)}{1+ hs f(\nu)}.
\end{equation*}
In order to take advantage of \eqref{linearity}, we rewrite $T_h^{}$ as 
\begin{equation}\label{rewritingT}
T_h^{}(\nu) = \Big (1 - h \! \sum_{i \in S^\circ} \varrho_i^{} \Big ) \Tsel{h}{}(\nu) + h \! \sum_{i \in S^\circ} \varrho_i^{} \Tsel{h}{}(\nu)_{C_i}^{} \otimes \nu_{D_i}^{} + \cO(h^2),
\end{equation}
where the implied constant is uniform in $\nu$; note that $f$ agrees on all explicit summands. Setting
\begin{equation}\label{Ttilde}
\widetilde{T}_h^{} (\nu) \defeq\Big (1 - h \! \sum_{i \in S^\circ} \varrho_i^{} \Big ) \Tsel{h}{}(\nu) +  h \! \sum_{i \in S^\circ} \varrho_i^{} \Tsel{h}{}(\nu)_{C_i}^{} \otimes \nu_{D_i}^{},
\end{equation}
we see that 
\begin{equation}\label{eulerscheme}
\lim_{h \to 0} \widetilde{T}^{\lfloor \frac{t}{h} \rfloor}_h (\omega_0^{}) = \omega_t^{},
\end{equation}
again uniformly in $\omega_0^{}$ and locally uniformly in $t$,
which is an easy consequence of Eq.~\eqref{eulerschemeoriginal} and the fact that $\widetilde T_h^{} = T_h^{} + \cO(h^2)$. Indeed, since $T_h^{}$ maps probability measures into probability measures, which have norm 1 by definition, we have, for all $h>0$, $m \in \NN$, and some $K >0$  that
\begin{equation*}
\|\widetilde T_h^m - T_h^m \| = \| (T_h^{} + \cO(h^2))^m - T_h^m\| \leqslant \sum_{j = 1}^m \binom{m}{j} K^j h^{2j} = (1 + Kh^2)^m - 1.
\end{equation*}
Thus, setting $m = \lfloor \frac{t}{h} \rfloor$ and assuming $t < \tau < \infty$, we get 
\begin{equation*}
\limsup_{h \to 0} \|\widetilde T_h^{\lfloor \frac{t}{h} \rfloor} - T_h^{\lfloor \frac{t}{h} \rfloor} \| \leqslant   \limsup_{h \to 0} (1 + Kh^2)^{\lfloor \frac{t}{h} \rfloor} - 1 \leqslant \limsup_{h \to 0} (1 + Kh^2)^{\frac{\tau}{h^2}h} - 1 = \limsup_{h \to 0} \ee^{K \tau h} - 1 = 0,
\end{equation*}
thus proving Eq.~\eqref{eulerscheme}.
As mentioned above, the advantage of working with $\widetilde T_h^{}$ rather than $T_h^{}$ is that $f$ agrees on all summands so that Eq.~\eqref{linearity} facilitates the computation of higher powers. 

We will be guided by the following genealogical interpretation of the approximate Euler scheme $\widetilde T_h^{}$. Assume that we sample a random individual, whom we will call `Bob', from the population $\omega_t^{}$, and determine his type. We distinguish two possibilities.
\begin{enumerate}[label=(\roman*)]
\item \label{item:noreco}
With (approximate) probability $1 - h \! \sum_{i \in S^\circ} \varrho_i^{}$, Bob has a \emph{single} ancestor at time $t-h$. In this case, only selection plays a role and Bob's type is approximately (that is, up to order $h^2$) distributed according to 
$\Tsel{h}{}(\omega_{t-h}^{})$.
\item \label{item:yesreco}
With probability $h \varrho_i^{}$ for all $i \in S^\circ$, Bob has \emph{two} ancestors, whose sequences performed a single crossover at site $i$. In this case (due to random mating), the alleles at the sites of Bob's sequence that are contained in $C_i$ are independent\footnote{That different ancestors imply independence is due to the deterministic limit considered here, more precisely to the absence of coalescence of ancestral lineages; cf. Section~\ref{sec:outlook}.} of those at the sites contained in $D_i$. As only the $i$-head, $C_i$, is affected by selection (since it contains the selected site), Bob's type is approximately distributed according to $\Tsel{h}{}(\omega_{t-h}^{})^{}_{C_i} \otimes \omega_{D_i,t-h}^{}$.
\end{enumerate}
All other possibilities, such as multiple crossovers, occur with a probability of order $h^2$ and are thus neglected.

To compute powers of $\widetilde T_h^{}$ then means to trace Bob's ancestry further into the past, eventually expressing the distribution of his type in terms of the initial type distribution of the population. As a first step, we consider the situation that until backward time $mh$, $m \in \NN$, no recombination occurred in Bob's ancestry so that his type is (approximately) distributed according to $\Tsel{h}{m}(\omega_{t-mh})$. Now, we look one step further into the past, distinguishing whether a crossover has occurred, which happens with the same probability as above.
\begin{enumerate}[label=(\alph*)]
\item \label{item:selnoreco}
If no recombination occurred in this step either, Bob's type is distributed as $\Tsel{h}{m+1}(\omega_{t - (m+1)h}^{})$. 
\item \label{item:selyesreco}
However, if recombination \emph{did} occur, say, at site $i$, then the head and tail are independent, as in case~\ref{item:yesreco}. Moreover, as only the site $i_\bullet^{} \in C_i$ is under selection, selection only acts along the ancestry associated with $C_i$. The tail $D_i$ on the other hand is contributed by a different, independent ancestral individual, replacing the original instance of $D_i$ that would otherwise hitchhike along with $C_i$. Thus, the distribution of Bob's type is in this case given by 
$
\Tsel{h}{m+1}(\omega_{t - (m+1)h}^{})_{C_i}^{} \otimes \omega_{D_i,t - (m+1)h}^{}.
$
\end{enumerate}
\begin{remark}\label{rem:physical_link}
Strictly speaking, for any given individual, its tail is obviously always linked to the corresponding head along which it hitchhikes. The statement in \ref{item:selyesreco} is therefore to be understood in a purely statistical sense; that is, it only refers to type \emph{distributions}, not individuals. Put differently, the head and tail are \emph{physically linked}, as they  belong to the same individual, looking \emph{forward} in time. However, they are \emph{statistically unlinked} due to having independent ancestors \emph{backward} in time.
\end{remark}
Formally, looking one step further into the past amounts to approximating $\Tsel{h}{m}(\omega_{t - mh}^{})$ by $\Tsel{h}{m} \big (\widetilde T_h^{} (\omega_{t - (m+1)h}^{}) \big )$. Thus, the preceding discussion can be formalised as follows.
\begin{lemma}\label{restrictionislinear}
For all $\nu \in \cP(X)$ and all $m \in \NN$,
\begin{equation*}
\Tsel{h}{m}\big (\widetilde{T}_h^{}(\nu) \big ) = \Big ( 1 - h \! \sum_{i \in S^\circ} \varrho_i^{} \Big ) \Tsel{h}{m+1}(\nu) + h \! \sum_{i \in S^\circ} \varrho_i^{}  \Tsel{h}{m+1}(\nu)^{}_{C_i} \otimes \nu^{}_{D_i}.
\end{equation*}
\end{lemma}
\begin{proof}
The proof rests on our earlier observation~\eqref{linearity} that $\Tsel{h}{}$ is linear on level sets of $f$, together with the fact that for $\mu$ and $\mu' \in \cP(X)$ and all $i \in S^\circ$, we have
\begin{equation}\label{secondobservation}
\Tsel{h}{} (\mu_{C_i}^{} \otimes \mu'_{D_i}) = \Tsel{h}{}(\mu)^{}_{C_i} \otimes \mu'_{D_i},
\end{equation}
which expresses that selection only acts on the head, as mentioned above.
To see this, note that  Eq.~\eqref{alternative} implies for all $x \in X$ that
\begin{equation}
\label{niceplay}
\begin{split}
F \big (\mu^{}_{C_i} \otimes \mu'_{D_i}\big ) (x) &= (1 - x_{i_\bullet^{}}^{})  \mu(x^{}_{C_i},\ast) \mu' (\ast,x^{}_{D_i}) 
= (F \mu)(x^{}_{C_i},\ast) \mu' (\ast,x^{}_{D_i}) \\
& = \big ( (F \mu)^{}_{C_i} \otimes \mu'_{D_i} \big ) (x),
\end{split}
\end{equation}
where we assumed $i_\bullet^{} < i$ without loss of generality. 
Recalling that
$\Psisel (\mu) = s \big ( F(\mu) - f(\mu) \mu \big)$,
 we also see that
\begin{equation}\label{psiselproduct}
\begin{split}
\Psisel(\mu_{C_i}^{} \otimes \mu'_{D_i}) & = s \big ( F(\mu_{C_i}^{} \otimes \mu'_{D_i}) - f(\mu_{C_i}^{} \otimes \mu'_{D_i})\mu_{C_i}^{} \otimes \mu'_{D_i} \big )  \\
&= s \big (F(\mu) - f(\mu)\mu \big)_{C_i}^{} \otimes \mu'_{D_i}
= \Psisel(\mu)_{C_i}^{} \otimes \mu'_{D_i},
\end{split}
\end{equation}
where we have used \eqref{niceplay} together with $f(\mu^{}_{C_i} \otimes \mu'_{D_i})=f(\mu)$ and the bilinearity of $\otimes$ in the second step. This also proves Eq.~\eqref{secondobservation} via $\Tsel{h}{}(\mu_{C_i}^{} \otimes \mu'_{D_i}) = \big (\mu + h \Psisel(\mu) \big )_{C_i}^{} \otimes \mu'_{D_i}$.

We now proceed via induction; for $m=0$, the statement reduces to the definition of $\widetilde T_h^{}$ (see Eq.~\eqref{Ttilde}). Assuming $m \geqslant 1$, we compute
\[
\begin{split}
\Tsel{h}{m}  \big (\widetilde T_h(\nu)  \big ) & = \Tsel{h}{} \big [ \Tsel{h}{m-1} \big ( \widetilde T_h(\nu) \big  ) \big ] \\
& = \Tsel{h}{}  \Big [ \Big (1-h \! \sum_{i \in S^\circ} \varrho^{}_i \Big ) \Tsel{h}{m}(\nu) + h \! \sum_{i \in S^\circ} \varrho^{}_i \Tsel{h}{m}(\nu)^{}_{C_i} \otimes \nu^{}_{D_i} \Big ] \\
& = \Big (1-h \! \sum_{i \in S^\circ} \varrho^{}_i \Big ) \Tsel{h}{}(\Tsel{h}{m}(\nu)) + h \! \sum_{i \in S^\circ} \varrho^{}_i \Tsel{h}{} (\Tsel{h}{m}(\nu)^{}_{C_i} \otimes \nu^{}_{D_i}) \\
& = \Big (1-h \! \sum_{i \in S^\circ} \varrho^{}_i \Big ) \Tsel{h}{m+1}(\nu) + h \! \sum_{i \in S^\circ} \varrho^{}_i  \Tsel{h}{m+1}(\nu) ^{}_{C_i} \otimes \nu^{}_{D_i}.
\end{split}
\]
Here, we applied the induction hypothesis in the second step, while the third step takes advantage of Eq.~\eqref{linearity} and the last step follows from Eq.~\eqref{secondobservation} with $\mu = \Tsel{h}{m}(\nu)$ and $\mu' = \nu$.
\end{proof}

One consequence of the previous lemma and the definition of $\widetilde T_h^{}$ is that any power of $\widetilde T_h^{}(\nu)$ can be expressed as a convex combination of terms of the form
\begin{equation}\label{recombined}
\Tsel{h}{m_1^{}}(\nu)_{A_1}^{} \otimes \ldots \otimes \Tsel{h}{m_r^{}}(\nu)_{A_r}^{},
\end{equation}
where $\{A_1, \ldots, A_r\}$  is an (interval) partition of $S$; this partition, along with $m_1^{},\ldots,m_r^{}$,  is independent of $\nu$.
We want to think about this in genealogical terms, continuing the discussion preceding the lemma. First, we only consider the effect of recombination and let $s=0$, which implies $\Tsel{h}{} = \id$. Then, \eqref{recombined} reduces to the product  $\nu_{A_1}^{} \otimes \ldots \otimes \nu_{A_r}^{}$ of marginals. This is reminiscent of the representation of the solution of the recombination equation via a \emph{partitioning process} $\varSigma = (\varSigma_{mh})_{m \in \NN_0}$ in discrete time, as in \citet{haldane,recoreview}. The partitioning process is a Markov chain on the (interval) partitions of $S$ and should be thought of as running \emph{backward} in time, describing Bob's ancestry in the following way. Each block $A$ of $\varSigma_{mh}$ corresponds to an  independent ancestor, alive at time $t - mh$ (when the type distribution in the population was given by $\nu$), which contributes to Bob's genome the alleles at the sites in $A$; hence the product of the corresponding marginals. Motivated by \ref{item:noreco} and \ref{item:yesreco} above, each block undergoes the following transitions, independently of all others.
\begin{enumerate}[label=(\roman*')]
\item \label{item:ppnoreco}
With probability $1 - h \! \sum_{i \in S^\circ} \varrho_i^{}$, $A$ remains unchanged.
\item \label{item:ppyesreco}
With probability $h \varrho_i^{}$ for all $i \in S^\circ$, the block $A$ is replaced by  two blocks $A \cap C_i$ and $A \cap D_i$ if they are both nonempty; if one of them is empty, the other is $A$ (since $C_i \cup D_i = S$) and the event is silent.
\end{enumerate}
At least heuristically, this argumentation leads to the following stochastic interpretation of the powers of $\widetilde T_h^{}$ for $s = 0$:
\begin{equation}\label{TtildePPrep}
\widetilde T_h^m (\nu) = \EE \Big [ \bigotimes_{A \in \varSigma_{hm}} \nu_A^{} \mid \varSigma_0 = \{S\} \Big].
\end{equation}
For $s > 0$, we describe the ancestral structure by a \emph{labelled}
 partitioning process $\widetilde \varSigma = (\widetilde \varSigma_{mh})_{m \in \NN_0}$ instead; it is related to the unlabelled version
  $\varSigma$ by associating to each block $A$ a label $\theta_A^{}$, which we also call
   the (\emph{selective}) \emph{age} of that block. That is, the amount of time that the sites
    in $A$ hitchhike along with the selected site.  
       Conditional on the evolution of the blocks, their associated labels evolve as follows.
\begin{enumerate}[label=(\alph*')]
\item 
In case of a transition of type \ref{item:ppnoreco}, $\theta_A^{}$ is incremented by $h$.
\item \label{item:ppselyesreco}
In case of a transition of type \ref{item:ppyesreco}, we always set $\theta_{C_i \cap A}^{} \defeq \theta_A^{}$ and $\theta_{D_i \cap A}^{} \defeq 0$, even if the transition is silent on the level of blocks. In particular, whenever $D_i$ is separated as a whole from $i_\bullet^{}$, its label is reset to $0$.
\end{enumerate}
See also \ref{item:selyesreco} and Remark~\ref{rem:physical_link} for the genealogical interpretation of the resetting that occurs in the context of \ref{item:ppselyesreco}.
In analogy with the representation \eqref{TtildePPrep} of $\widetilde T_h^m$ for $s = 0$, we expect for $s > 0$ that
\begin{equation}\label{lpartprocessrep}
\widetilde T_h^m (\nu) = \EE \Big [ \bigotimes_{A \in \Sigma_{hm}} \Tsel{h}{\theta_A^{} / h} (\nu)_A^{} \mid \Sigma_0 = \{S\}, \theta_S^{} = 0 \Big].
\end{equation}
Before making this rigorous, we  provide a visualisation via the (discrete) ancestral initiation graph (AIG).

\begin{definition}\label{hAIG}
The $h$-AIG of length $kh$ is a random graph with labelled leaves, which is constructed recursively and from right to left as follows. For $k=0$, it consists of a single root, which coincides with a single leaf and carries the label $0$. For $k >0$, we construct the AIG of length $kh$ from an AIG of length $(k-1)h$ by attaching to each leaf (with label $mh$)
\begin{itemize}
\item 
an edge
\begin{center}
\includegraphics{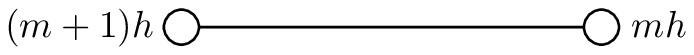}
\end{center}
of length $h$ with a leaf labelled $(m+1)h$, with probability $1 - h \! \sum_{i \in S^\circ} \varrho_i^{}$, 
\item
an $i$-splitting
\begin{center}
\includegraphics{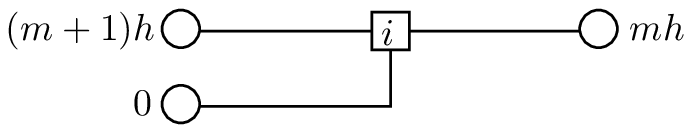}
\end{center}
again of length $h$, with probability $h \varrho_i^{}$  for $i \in S^\circ$.
The upper line carries a leaf with label $(m+1)h$, while the lower line carries a leaf with label 0.  
\end{itemize}
The previous leaves keep their labels but turn into internal vertices.
We refer to the labels of the vertices as their (selective) \emph{ages}.
\end{definition}
The root of the $h$-AIG should be thought of as a randomly chosen member of the current population, and the remaining vertices represent its ancestry at the time points $mh$ before the present. More precisely, we keep track of the sites that each of these vertices is ancestral to. By convention, the root itself is ancestral to $S$. If a vertex is connected to another one via a single edge, the left vertex is ancestral to the same set of sites as the right one. In case of a splitting at site $i$, if the right vertex is ancestral to $A \subseteq S$, then the upper left vertex is ancestral to $A \cap C_i$ while the lower left vertex is ancestral to $A \cap D_i$; if either of these sets is empty, we call the corresponding vertex \emph{nonancestral}.  Pruning  away the subgraph spanned by the non-ancestral vertices results in the (true) ancestry (of the root). Finally, for any nonempty $B \subseteq S$, we call the subgraph spanned by all vertices ancestral to $B' \subseteq S$ with $B' \cap B \neq \varnothing$ the \emph{ancestry} of $B$. In particular, if $B = \{i\}$ is a singleton, we call the resulting sequence of vertices the \emph{ancestral line} of site $i$.

\begin{remark}
One may regard the $h$-AIG as a Markov chain in discrete time. Then, the ancestry (as defined above) can be interpreted as an embedding of the aforementioned labelled partitioning process into the $h$-AIG.
\end{remark}

\begin{figure}[t]
\includegraphics[width = 0.65\textwidth]{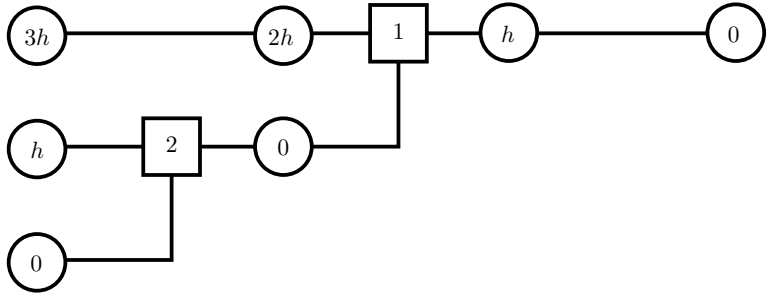}
\caption{\label{haigexample}
A realisation of the $h$-AIG of length $3h$ with $i_\bullet=1$; its root, leaves, and internal vertices are visualised by circles with inscribed labels. The squares correspond to splittings, and the inscribed numbers indicate the sites where the splittings occur.  
}
\end{figure}

Accordingly, given a realisation of the $h$-AIG (see Fig.~\ref{haigexample} for an example) and an initial type distribution $\nu$, we construct the (random) type of the root as follows.
\begin{enumerate}[label=(\arabic*)]
\item[(I)] \label{item:ancsampling}
Assign types to the leaves according to their ages; if a leaf has age $mh$, sample its type according to $\Tsel{h}{m}(\nu)$. 
\item[(II)]
Starting at the leaves, propagate the types  from left to right along the lines of the graph. Whenever two lines are joined in an $i$-splitting, the  type of the descendant is obtained by joining the alleles at the sites in $C_i^{}$ of the upper line with the alleles at the sites in $D_i^{}$ of the lower line.
\item[(III)]
Eventually, this procedure assigns a (random) type to the root; we call it the \emph{type delivered} by the $h$-AIG from the initial distribution $\nu$.
\end{enumerate}
In this sense, the $h$-AIG can be thought of as a random (discrete) flow on $\cP(X)$. In the mean, we recover the discrete approximation of the flow semigroup of the selection-recombination equation generated by $\widetilde T_h^{}$. 

\begin{prop}\label{hAIGdelivery}
Let $\nu \in \cP(X)$ and $k \in \NN$. 
The type delivered by an $h$-\textnormal{AIG} of length $kh$ from the initial distribution $\nu$ is distributed as $\widetilde{T}_h^k(\nu)$.
\end{prop}

\begin{proof}
For $k = 0$, there is nothing to show. Assume that the statement holds for $k \geqslant 0$. 

Recall that when passing from the $h$-AIG of length $kh$ to the $h$-AIG of length $(k+1)h$, we attach independently to each leaf an edge with probability 
$1 - h \! \sum_{i \in S^\circ}\varrho_i^{}$. If that leaf has age $mh$, then, by construction, the new leaf has age $(m+1)h$, and its type has distribution $\Tsel{h}{m+1}(\nu)$. With probability $h\varrho_i^{}$, an $i$-splitting is attached. Recall that the upper leaf (which contributes sites $1$ to $i$) has age $(m+1)h$, while the lower leaf (which contributes sites $i + 1$ to $n$) has age $0$. Therefore, the resulting type has distribution $\Tsel{h}{m+1}(\nu)_{C_i}^{} \otimes \nu_{D_i}^{}$.

To summarise, passing from the $h$-AIG of length $kh$ to the $h$-AIG of length $(k+1)h$ has the same effect as replacing the distribution $\Tsel{h}{m}(\nu)$ of a leaf of age $mh$ by 
\begin{equation*}
\Big (1 - h \! \sum_{i \in S^\circ}\varrho_i^{} \Big) \Tsel{h}{m+1}(\nu) + h \! \sum_{i \in S^\circ} \varrho_i^{} \, \Tsel{h}{m+1}(\nu)_{C_i}^{} \otimes \nu_{D_i}^{} = \Tsel{h}{m}(\widetilde T_h^{}(\nu)),
\end{equation*}
where the equality holds by Lemma~\ref{restrictionislinear}. The net effect of passing from length $kh$ to length $(k+1)h$ in the $h$-AIG thus amounts to replacing the initial type distribution $\nu$ by $\widetilde T_h^{}(\nu)$.  By the induction hypothesis, this concludes the proof.
\end{proof}

Let us take a moment to appreciate what we have accomplished so far. We have constructed a discrete approximation of the flow associated with Eq.~\eqref{recoselequation} that can be realised via a graphical construction describing the genealogical structure of a sample due to recombination. We next let the step size $h \to 0$, which will yield a graphical representation of the \emph{exact} solution $\omega$.

To this end, we recall Eq.~\eqref{eulerscheme}, consider  the $h$-AIG of length $\lfloor t / h \rfloor$, and let $h \to 0$. First, note that the vertices, which represent the (potential) ancestors of the root, move closer and closer together as $h \to 0$ so that, in the limit, (potentially) ancestral lines turn from  sequences of vertices into continuous lines.  In between two vertices, an $i$-splitting occurs with probability $h\varrho_i^{}$, meaning that the distance between any vertex and the next $i$-splitting (in either direction of time) is distributed as $h\cG_h$ where $\cG_h$ is geometrically distributed with success probability $h \varrho_i^{}$. More precisely, $\cG_h$ counts the number of failures up to (and not including) the first success, where a success corresponds to a splitting between two successive vertices.  Thus, in the limit, the distance between any point and the next $i$-splitting will be exponentially distributed with parameter $\varrho_i^{}$. Therefore, we define the (continuous) ancestral initiation graph (AIG) as follows; see Figure~\ref{aigpartition} for an illustration.

\begin{definition}\label{aig}
The AIG (of length $t$) is a random graph of length $t$ that is grown from right to left, starting with a single line emanating from its root. For $i \in S^\circ$, each line is affected by $i$-splittings at an exponential rate $\varrho_i^{}$. Each leftmost point is called a \emph{leaf} and is labelled by the length (or age) $\theta$ of the line it is attached to, measured from the point where that line has split off; if it has not split off from any line (as is the case for the top line), it has length $t$. We will abbreviate the AIG of length $t$ by $\Gamma_t^{}$.
\hfill $\diamondsuit$
\end{definition}

\begin{figure}
\includegraphics[width = 0.8 \textwidth]{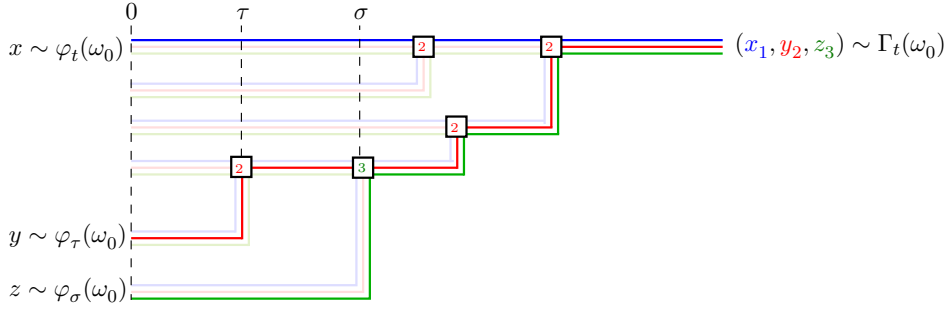}
\caption{A realisation of the AIG with $i_\bullet=1$, along with the partitioning of the root sequence into ancestral lines. Sites 1,2, and 3 are blue, red, and green; ancestral (nonancestral) material is in opaque (pale) colour. }
\label{aigpartition}
\end{figure}

In the following, we denote the AIG of length $t$ by $\Gamma_t$. As in the discrete case, it naturally embeds into a graph-valued Markov process $\Gamma = (\Gamma_t)_{t \geqslant 0}$, which we simply call the AIG (without specifying its length); it can be interpreted as a random flow on $\cP(X)$ via the following sampling procedure, which is derived from the corresponding procedure in the discrete case by replacing (I) above by
\begin{enumerate}[label=(\arabic*')]
\item[(I')]
Assign types to the leaves according to their ages; if a leaf has age $\theta$, sample its type according to $\varphi_\theta^{}(\nu)$ (recall from Theorem~\ref{recursion} that $\varphi$ denotes the flow of the pure selection equation).
\end{enumerate}
In line with this interpretation, we denote by $\Gamma_t(\nu)$ the distribution of the type delivered by the AIG of length $t$ from the initial distribution $\nu$. With this, the findings of this section can be succinctly summarised by the equality
\begin{equation}\label{aigdeliveryeq}
\omega_t^{} = \EE [ \Gamma_t(\omega_0^{})],
\end{equation}
where the expectation is taken with respect to all realisations of $\Gamma$.

\begin{remark}\label{connectiontoessASRG}
Before we embark on the proof of Theorem~\ref{recursion}, let us take a moment to relate the AIG to the constructions introduced in~\citet{Selrek}.
\begin{enumerate}[label = (\Alph*)]
\item \label{item:timetolinecount}
So far, we did not consider the ancestral structure due to selection, which we made up for by sampling the type of each leaf of the AIG from $\varphi_\theta^{} (\omega^{}_0)$, according to its (selective) age $\theta$, rather than from $\omega_0^{}$. Alternatively, we may take advantage of the well-known duality of the pure selection equation and the \emph{ancestral selection graph} (ASG). For each leaf, we keep track of the random number of potential ancestors within an associated independent copy of the ASG. Independently for each leaf and in the absence of splitting, this number grows according to a Yule process with (binary) branching rate $s$, starting with a single potential ancestor associated to the root. Upon a splitting, the leaf of the top line inherits the number of ancestors, while the leaf of the lower line starts from scratch with a single ancestor. 

In the sampling step, a type is assigned to a leaf which carries, say, $r$ potential ancestors, via the following two-step procedure. First, a type is sampled according to $\omega_0^{}$ for each potential ancestor. These then compete for the true ancestry, where the fit individuals prevail over the unfit ones. That is, if there is at least one fit individual, the true ancestral type is chosen uniformly among the fit ones, and uniformly among all samples, otherwise. To summarise: if a leaf carries $r$ potential ancestors, its type is distributed according to
\begin{equation*}
(1 - f(\omega^{}_0))^r d(\omega^{}_0) + \big ( 1 - (1-f(\omega^{}_0))^r \big) b(\omega^{}_0)
\end{equation*}
Finally, the types are propagated through the graph as before. 
\item
The blocks of the partitioning process associated with the true ancestry of an individual under recombination, illustrated by the opaque lines in Fig.~\ref{aigpartition}, can be decorated with the number of potential ancestors under selection in the sense of \ref{item:timetolinecount}. This gives the \emph{weighted} partitioning process in the sense of \citet[Section 7.1]{Selrek}.
\item
Rather than just keeping track of the number of potential ancestors for each leaf, we can instead keep track of the full graphical representation of the ASG at any given time; this leads to the \emph{essential} ASRG discussed in \citet[Section 6]{Selrek}.
\item
Similarly to how the $h$-AIG can be seen as a graphical encoding of the discrete labelled partitioning process, the AIG can be seen as a graphical encoding of its analogue in continuous time, as illustrated by the opaque lines in Fig.~\ref{aigpartition}. By the single-crossover assumption, one can succinctly encode this partitioning process  as a vector-valued process $(\widetilde \Theta_t)_{t \geqslant 0} = (\tilde \theta_1^{},\ldots,\tilde \theta_n^{})_{t \geqslant 0}$, where $\tilde \theta_{i_\bullet^{},t}^{}=t$ is the time selection has acted on the selected site; whereas $\tilde \theta_{i,t}^{}$, $i \in S^\circ$, takes values in $\RR_{\geqslant 0} \cup \{\Delta\})$. More precisely, for $i \in S^\circ$, $\tilde \theta_{i,0}^{}=\Delta$; for $t>0$, $\tilde \theta_{i,t}^{}\in \RR_{\geqslant 0}$ is the time since the last splitting event on the ancestral line of site $i$ that separated it from $i_\bullet^{}$, whereas $\tilde \theta_{i,t}^{}=\Delta$ indicates that no such event has occurred yet. 
The time evolution of $(\widetilde \Theta_t)_{t \geqslant 0}$ is given by an independent collection of \emph{initiation processes}, which are named so because they describe the `initiation' of a new `selection epoch' by every splitting event; see \citet[Section~7.3]{Selrek}.
\end{enumerate}
\hfill $\diamondsuit$ 
\end{remark}

\subsection*{Proof of Theorem~\ref{recursion}} In line with the assumption of the theorem, we restrict ourselves to the case $i_\bullet^{} = 1$. We will now use the stochastic representation~\eqref{aigdeliveryeq} of $\omega$ via the AIG to prove the theorem. In perfect analogy with the recursion in Theorem~\ref{recursion} and for all $0 \leqslant k < n$, we define the \emph{truncated} AIGs, $\Gamma^{(k)}$, by ignoring all $i$-splittings in $\Gamma$ for $i > k+1$.  In particular, $\Gamma^{(0)}$ consists only  of a single edge, and $\Gamma^{(n-1)} = \Gamma$ is the full AIG.  As for the full AIG, we denote the type delivered by $\Gamma_t^{(k)}$ from the initial distribution $\nu$ by $\Gamma_t^{(k)}(\nu)$. In analogy to Eq.~\eqref{aigdeliveryeq}, we have
\begin{equation}\label{stochasticrep}
\omega_t^{(k)} = \EE \big [\Gamma_t^{(k)} (\omega_0^{}) \big].
\end{equation}
The key insight is that $\Gamma^{(k-1)}$ naturally embeds into $\Gamma^{(k)}$ by ignoring all $k + 1$-splittings in $\Gamma^{(k)}$; see Fig.~\ref{aigexample}. 
\begin{figure}[t]
\includegraphics[width = 0.8\textwidth]{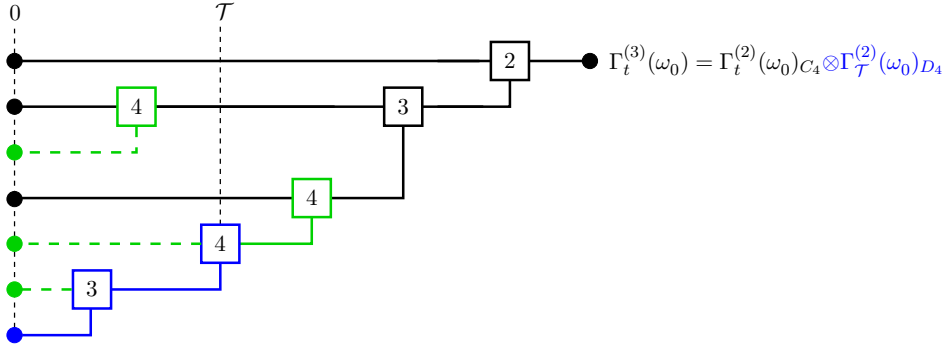}
\caption{\label{aigexample}
A realisation of $\Gamma_t^{} = \Gamma_t^{(3)}$ (with $n=4$ and $i_\bullet=1$). At its root, the distribution $\Gamma_t(\omega^{}_0)$ of the delivered type; here, it decomposes into the marginals with respect to $C_4$ and $D_4$. Ignoring the green and blue elements yields a realisation of $\Gamma_t^{(2)}$. Lines that do not carry ancestral material are dashed. The part of the ancestral line of $D_4$ on which no further $4$-splittings occur is  drawn in blue (the last such splitting  happens at (forward) time $\cT$).}
\end{figure}
Moreover, we can decompose $\Gamma^{(k)}$ into the ancestries of $C_{k+1}$ and $D_{k+1}$ and make the following two observations. First, the ancestry of $C_{k+1}$ is not affected by $k+1$-splittings. This is because  the newly attached lines in a $k+1$-splitting belong to  $D_{k+1}$ and are nonancestral to $C_{k+1}$. In short, this implies that
\begin{equation*}
\Gamma_t^{(k)}(\omega_0^{})_{C_{k+1}} = \Gamma_t^{(k-1)}(\omega_0^{})_{C_{k+1}}
\end{equation*}
always, in line with the marginalisation consistency discussed in~\citet[Appendix A]{Selrek}.

Next, we consider the ancestry of $D_{k+1}$, which consists only of the ancestral line of site $k+1$, due to the absence of $i$-splittings for $i > k+1$. There are two possibilities. With probability $\ee^{-\varrho_{k+1}^{} t}$, no $k+1$-splitting occurs on this line so that it travels together with the ancestral line of site $k$. In this case, we have
\begin{equation*}
\Gamma_t^{(k)}(\omega_0^{}) = \Gamma_t^{(k-1)}(\omega_0^{}).
\end{equation*}
On the other hand, if a $k+1$-splitting \emph{did} occur on the ancestral line of site $k+1$, then the alleles at $C_{k+1}$ and $D_{k+1}$ are provided by different, independent ancestors, whence
\begin{equation*}
\Gamma_t^{(k)}(\omega_0^{}) = \Gamma_t^{(k)}(\omega_0^{})_{C_{k+1}} \otimes \Gamma_t^{(k)}(\omega_0^{})_{D_{k+1}}
= \Gamma_t^{(k-1)}(\omega_0^{})_{C_{k+1}} \otimes \Gamma_t^{(k)}(\omega_0^{})_{D_{k+1}}.
\end{equation*}
Moreover, in the absence of $k+1$-splittings on the ancestral line of site $k+1$, the age of this line evolves as though it was part of an independent copy $\widetilde{\Gamma}$ of $\Gamma$. Thus, denoting by $\cT$ the distance from the leaf to the \emph{leftmost} $k+1$-splitting, which is exponentially distributed with mean $1 / \varrho_{k+1}^{}$ conditional on being $\leqslant t$,  we have
$\Gamma_t^{(k)}(\omega_0^{})_{D_{k+1}} = \widetilde \Gamma_\cT^{(k-1)}(\omega_0^{})_{D_{k+1}}$ in distribution.

To summarise, with $\cE \sim \text{Exp}(\varrho_{k+1}^{})$ (where $\exponential(\lambda)$ denotes the exponential distribution with parameter $\lambda$), we obtain 
\begin{equation*}
\Gamma_t^{(k)}(\omega_0^{}) = \one_{\cE \geqslant t}^{} \Gamma_t^{(k-1)} (\omega_0^{}) + \one_{\cE < t}^{} \Gamma_t^{(k-1)}(\omega_0^{})_{C_{k+1}} \otimes \widetilde \Gamma_\cE^{(k-1)}(\omega_0^{})_{D_{k+1}}
\end{equation*}
in distribution,
where  $\one$ denotes the indicator function. Taking the expectation on both sides  yields
\begin{equation*}
\begin{split}
\omega_t^{(k)} &= \EE \big [\Gamma_t^{(k)} (\omega^{}_0) \big ] = 
\EE \big [\one_{\cE \geqslant t}^{} \Gamma_t^{(k-1)} (\omega_0^{}) + \one_{\cE < t}^{} \Gamma_t^{(k-1)}(\omega_0^{})_{C_{k+1}} \otimes \widetilde \Gamma_\cE^{(k-1)}(\omega_0^{})_{D_{k+1}}  \big] \\
&= \ee^{-\varrho_{k+1}^{} t} \EE \big [\Gamma_t^{(k-1)} (\omega^{}_0) \big ] + \EE \big [\Gamma_t^{(k-1)}(\omega_0^{}) \big ]_{C_{k+1}}^{} \otimes \int_0^t \varrho^{}_{k+1} \ee^{-\varrho_{k+1}^{} \tau} \EE \big [ \widetilde \Gamma_\tau^{(k-1)} (\omega_0^{}) \big ]^{}_{D_{k+1}} \dd \tau \\
&= \ee^{-\varrho_{k+1}^{} t} \omega_t^{(k-1)} + \omega_{C_{k+1},t}^{(k-1)} \otimes \int_0^t \varrho_{k+1}^{} \ee^{-\varrho_{k+1}^{} \tau} \omega_{D_{k+1},\tau}^{(k-1)} \dd \tau,
\end{split}
\end{equation*}
which proves Theorem~\ref{recursion}. \hfill \qed

We now present  a more algebraic argument that works directly at the level of the discrete flow. It can be viewed as a condensed version of the genealogical proof given above. Because some computations are similar to others that have already been carried out in detail elsewhere in the paper, we can give a streamlined argument here.
\subsection{An algebraic proof}
\label{subsec:purelyanalytic}
To reflect the recursion in Theorem~\ref{recursion}, we introduce $\widetilde T_h^{(k)}$, which is defined recursively as follows. We set $\widetilde T_h^{(0)} \defeq \Tsel{h}{}$, and, for $1 \leqslant k \leqslant n-1$,
\begin{equation}\label{deftildeThk}
\widetilde T_h^{(k)} (\nu) \defeq (1 - h \varrho_{k+1}^{}) \widetilde T_h^{(k-1)}(\nu) + h \varrho_{k+1}^{} \widetilde T_h^{(k-1)}(\nu)_{C_{k+1}} \otimes \nu_{D_{k+1}}^{}.
\end{equation}
Up to order $h^2$, $\widetilde T_h^{(k)}$ is just $\widetilde T_h^{}$ with $\varrho_{k+2}^{},\ldots,\varrho_n^{}$ set to $0$. Therefore, arguing exactly as for Eq.~\eqref{eulerscheme}, we have
\begin{equation}\label{eulerscheme2}
\lim_{h \to 0} \big (\widetilde T_h^{(k)} \big)^{\lfloor \frac{t}{h} \rfloor} (\omega_0^{}) = \omega_t^{(k)}.
\end{equation}

Recall that Eq.~\eqref{linearity} was instrumental in making sense of the powers of $\widetilde T_h^{}$. In analogy, we have 
\begin{equation}\label{linearity'}
\widetilde T_h^{(k)} (\alpha \mu + (1 - \alpha) \mu') = \alpha \widetilde T_h^{(k)} (\mu) + (1- \alpha) \widetilde T_h^{(k)}(\mu')
\end{equation}
for all $\mu,\mu' \in \cP(X)$ with $\mu_{C_{k+1}}^{} = \mu'_{C_{k+1}}$. This follows from Eq.~\eqref{linearity} together with the corresponding statement for the recombinators $\cR_i$ with $i \leqslant k+1$, which is obvious. In the same way, we see that
\begin{equation}\label{secondobservation'}
\widetilde T_h^{(k)} (\mu_{C_i}^{} \otimes \mu'_{D_i}) = \widetilde T_h^{(k)} (\mu)_{C_i} \otimes \mu'_{D_i}
\end{equation}
for $i \in S^\circ$, in analogy with \eqref{secondobservation}.
Proceeding exactly as in the proof of Lemma~\ref{restrictionislinear} with $\widetilde T_h^{(k-1)}$ in the place of $\Tsel{h}{}$, $\widetilde T_h^{(k)}$ in the place of $\widetilde T_h^{}$, as well as \eqref{linearity'} and \eqref{secondobservation'} replacing \eqref{linearity} and \eqref{secondobservation}, respectively,
 we obtain the following analogue of Lemma~\ref{restrictionislinear}: for all $m \in \NN$ and $\nu \in \cP(X)$,
\begin{equation}\label{singleline}
(\widetilde T_h^{(k-1)})^m \big ( \widetilde T_h^{(k)}(\nu) \big ) = (1 - h\varrho_{k+1}^{}) (\widetilde T_h^{(k-1)})^{m+1}(\nu) + h\varrho_{k+1}^{}(\widetilde T_h^{(k-1)})^{m+1}(\nu)_{C_{k+1}} \otimes \nu_{D_{k+1}}^{}.
\end{equation}

To evaluate the powers of $\widetilde T_h^{(k)}$, 
we now introduce the shorthands 
\begin{equation*}
A(\nu) \defeq \widetilde T_h^{(k-1)} (\nu) \quad \text{and} \quad \, B^{i,j}(\nu) \defeq \big (\widetilde T_h^{(k-1)}  \big )^i (\nu)_{C_{k+1}} \otimes \big (\widetilde T_h^{(k-1)}  \big )^j (\nu)_{D_{k+1}} \text{ for } i,j \in \NN_0.
\end{equation*}
Then, we can rewrite the definition of $\widetilde T_h^{(k)}$ as
\begin{equation}\label{stochastic_representation}
\widetilde T_h^{(k)} (\nu) =  \EE[\one_{\{\cB=0\}} A(\nu) + \one_{\{\cB=1\}} B^{1,0}(\nu)],
\end{equation}
where $\cB$ is a Bernoulli random variable with success probability $h \varrho_{k+1}^{}$. Furthermore, it is immediate from Eq.~\eqref{singleline} that 
\begin{equation}\label{action_Ai}
\begin{split}
A^i \big (\widetilde T_h^{(k)} (\nu) \big ) & = (1 - h \varrho_{k+1}^{}) A^{i+1}(\nu) + h \varrho_{k+1}^{} B^{i+1,0}(\nu)  \\ &= \EE[\one_{\{\cB=0\}} A^{i+1}(\nu) + \one_{\{\cB=1\}}B^{i+1,0}(\nu)]
\end{split}
\end{equation}
and
\begin{equation}\label{action_Bij}
\begin{split}
B^{i,j}\big (\widetilde T_h^{(k)} (\nu)\big ) & = (1 - h \varrho_{k+1}^{}) B^{i+1,j+1}(\nu) + h \varrho_{k+1}^{} B^{i+1,0}(\nu) \\ &= \EE[\one_{\{\cB=0\}} B^{i+1,j+1}(\nu) + \one_{\{\cB=1\}} B^{i+1,0}(\nu)],
\end{split}
\end{equation}
with $\cB$ as above. Starting from the stochastic representation \eqref{stochastic_representation} of $\widetilde T_h^{(k)}$ and applying Eqs.~\eqref{action_Ai} and \eqref{action_Bij} in an inductive manner, we see that 
\begin{equation}\label{powerThk}
  \big (\widetilde T_h^{(k)} \big)^m (\nu) = \EE[\cC_m] (\nu),
\end{equation}
where $(\cC_m)_{m \in \NN}$ is a Markov chain on the set of operators of the form $A^i$ or $B^{i,j}$ where $i \geqslant 1$ and $j \geqslant 0$. Its initial distribution is given by
\begin{equation*}
\PP (\cC_1 = A) = 1 - h \varrho_{k+1}^{} = 1 - \PP(\cC_1 = B^{1,0})
\end{equation*}
and the transition probabilities are
\begin{equation*}
\PP (\cC_{m+1} = A^{i + 1} \mid \cC_{m} = A^i ) = 1 - h \varrho_{k+1}^{} = 1 - \PP(\cC_{m+1} = B^{i+1,0} \mid \cC_m = A^i )
\end{equation*}
and
\begin{equation*}
\PP (\cC_{m+1} = B^{i+1,j+1} \mid \cC_m = B^{i,j}) = 1 - h \varrho_{k+1}^{} = 1 - \PP (\cC_{m+1} = B^{i+1,0} \mid \cC_m = B^{i,j} ).
\end{equation*}
Given an infinite sequence $\cB_1,\cB_2,\ldots$ of independent Bernoulli random variables with success probability $h\varrho_{k+1}$, we can construct a realisation of this Markov chain as follows. Set $\cC_1 \defeq \one_{\{\cB_1=0\}} A + \one_{\{\cB_1=1\}}B^{1,0}$. For $i > 1$, if $\cC_{i-1} = A^{i-1}$, set $\cC_{i} = \one_{\{\cB_{i} = 0\}}^{} A^{i} + \one_{\{\cB_{i} = 1\}}^{} B^{i,0}$. If $\cC_{i-1} = B^{i-1,j}$, set $\cC_{i} \defeq \one_{\{\cB_{i} = 0\}}^{} B^{i ,j + 1} +  \one_{\{\cB_{i} = 1\}}^{} B^{i ,0}$.

It is then clear that $\cC_m=A^m$ if  $\cB_1= \ldots = \cB_m=0$; otherwise $\cC_m=B^{m,\ell}$, where $\ell$ is such that $\cB_{m-\ell-1}=1, \cB_{m-\ell}=\cB_{m-\ell+1}= \ldots= \cB_m=0$, that is, $\ell$ is the number of 0's after the last 1. Since the number of failures before the first success, in both directions of time, is geometrically distributed with parameter $h \varrho_{k+1}^{}$, we finally see that 
\begin{equation*}
\big (\widetilde T_h^{(k)} \big)^m (\nu) = \EE \big [ \one_{\{\cG \geqslant m\}}^{} \big (\widetilde T_h^{(k-1)} \big)^m (\nu)  + \one_{\{\cG < m\}}^{}
 \big (\widetilde T_h^{(k-1)} \big)^m (\nu)_{C_{k+1}} \otimes  \big (\widetilde T_h^{(k-1)} \big)^{\cG} (\nu)_{D_{k+1}} \big ],
\end{equation*}
where $\cG$ is a geometric random variable with success probability $h \varrho_{k+1}^{}$. Setting $\nu = \omega_0^{}$, $m \defeq \lfloor t / h \rfloor$, letting $h \to 0$ and noting that $h \cG \to \cE \sim \text{Exp}(\varrho_{k+1}^{})$ in distribution, we see that
\begin{equation*}
\omega_t^{(k)} = \EE [\one_{\{\cE > t\}}^{} \omega_t^{(k-1)} + \one_{\{\cE \leqslant t\}}^{} \omega_{C_{k+1},t}^{(k-1)} \otimes  \omega_{D_{k+1},\cE}^{(k-1)} ],
\end{equation*}
which is the recursion stated in the theorem.
\qed




\section{Application: linkage disequilibrium in selective sweeps}
\label{sec:LDs}

We close by showing how our results can explain the effect of a selective sweep on the correlation or \emph{linkage disequilibrium} (LD) between two neutral sites.  A \emph{selective sweep} \citep{MaynardSmithHaigh} occurs when a new beneficial mutation at  $i_\bullet$ becomes prevalent in the population and thus also increases the frequency of the alleles at the neutral sites that were associated with the beneficial mutation when it arose; these alleles thus \emph{hitchhike} along with the beneficial mutation. Here, we consider the simplest scenario of two neutral sites $L$ and $R$ that are linked to $i_\bullet$, and drop the assumption that $i_\bullet = 1$. 
Following \citet{stephansonglangley}, we therefore take $S = \{i_\bullet,L,R\} = \{1,2,3\}$, where $i_\bullet \in \{1,2,3\}$ is given and $L,R \in S \setminus i_\bullet$ satisfy $L<R$; $L$ and $R$ denote the `left' and the `right' neutral site, respectively, see Fig.~\ref{SLR}. We then consider the LD or \emph{correlation function} between sites $L$ and $R$,

\begin{figure}[t]


\begin{center}
\includegraphics[width = 0.6\textwidth]{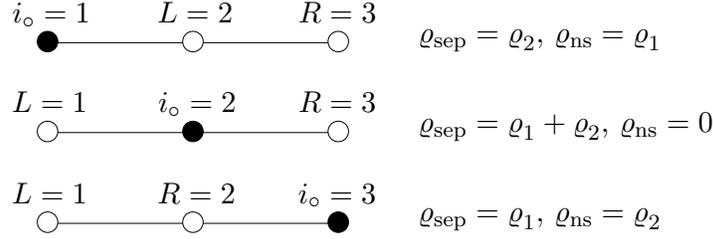}
\end{center}
\caption{\label{SLR}
The three cases $i_\bullet = 1, \, i_\bullet = 2$, and $i_\bullet = 3$. The selected site is represented by a bullet, the other two (neutral) sites by circles. For each arrangement, the values of $\rhosep$ and $\rhonosep$ are given in terms of $\varrho_L^{}$ and $\varrho_R^{}$. Note that $\rhonosep=0$ if $i_\bullet$ is in the middle. Therefore, the behaviour is fundamentally different in this case compared to when $i_\bullet$ is one of the outer  sites. See also Fig.~\ref{LDplot} for a comparison of the time-evolution of the LD in the case $\rhonosep = 0$ versus $\rhosep \neq 0$.}
\end{figure}

\begin{equation} \label{lddefinition}
\textnormal{Cor}(\omega_t) \defeq  \omega^{}_{\{L,R\},t}((1,1))  -  \omega^{}_{\{L\},t}(1) \omega^{}_{\{R\},t}(1).
\end{equation}
Our goal is to examine how the dynamics of LD is affected by the location of  $i_\bullet^{}$ relative to the neutral sites. In their Fig.~2, \citet{stephansonglangley} illustrate this dynamics by a numerical evaluation of their approximate solution and observe a  somewhat complicated behaviour, which remains a little mysterious. Some  hints at an explanation are given by \citet{PfaffelhuberLehnertStephan}, who, however, work in a  different setting: they consider a finite population in a  strong-selection approximation,  and focus on the LD at a fixed time close to the time of fixation. In what follows, we give a thorough discussion of the full dynamics in the law  of large numbers regime, both forward in time and in the genealogical sense.

As in the work cited above, we are interested in a single, rare beneficial mutation that is introduced into a homogeneous background. We model this by picking a single type \mbox{$x^{\text{mut}} \in \{x \in X : x^{}_{i_\bullet}=0\}$}  and set $\omega_0^{}(x^{\text{mut}}) \defeq \varepsilon$ (where $\varepsilon$ is a small positive number), together with $\omega^{}_0(x) \defeq 0$ for all $x \neq x^{\text{mut}}$ with $x^{}_{i_\bullet}=0$. More specifically, we choose $x^{\text{mut}}_L = x^{\text{mut}}_R = 1$, in line with \citet{stephansonglangley}, and adjust the remaining type frequencies such that 
\begin{itemize}
\item 
$\textnormal{Cor}(\omega_0) > 0$, and
\item
for $\varrho^{}_L=\varrho^{}_R=0$, one has $\frac{\dd}{\dd t} \textnormal{Cor}(\omega_t) |_{t=0} > 0$ (so that $(x^{}_L,x^{}_R)=(1,1)$  hitchhikes along with $x^{}_{i_\bullet}=0$).
\end{itemize}
The exact parameter values are given in the caption of Fig.~\ref{LDplot}.  

\begin{figure}[t]
\includegraphics[width=0.8\textwidth]{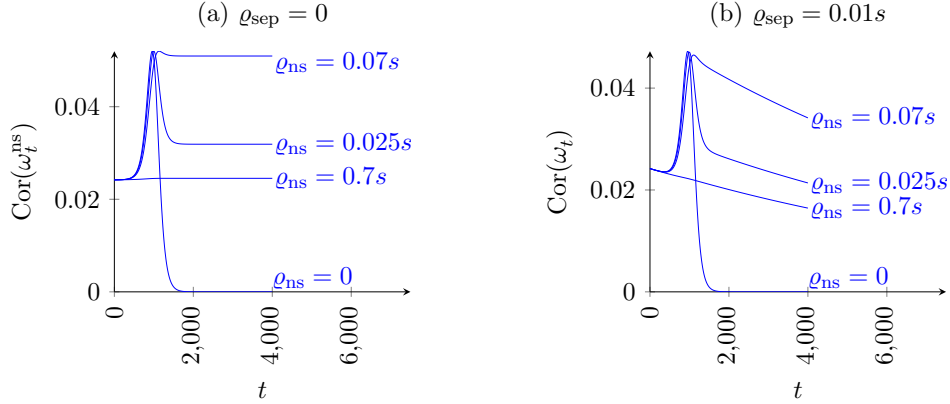}
\caption{\label{LDplot}
Time evolution of the correlation under selection and recombination as obtained by evaluating the solution formula from Theorem \ref{recursion}. In the left panel, recombination only separates the block $\{L,R\}$ from the selected site, but not $L$ and $R$ from each other. In the right panel, separating recombination is added. Parameters: $s = 10^{-2}$, $i_\bullet=1$,
$\omega_0(( 0,1,1) ) =  5 \cdot 10^{-5} = \varepsilon, \omega_0(( 1,1,1)) = 0.38995, \omega_0((1,0,1)) = 0.23, \omega_0((1,1,0)) = 0.2$ and $\omega_0((1,0,0)) =0.18$. 
}
\end{figure}

It is clear that, for $\varrho^{}_L=\varrho^{}_R=0$, there is an initial increase of LD due to hitchhiking, followed by an eventual decay to zero; compare the curve for $\varrho_\textnormal{ns} =0$ in Fig.~\ref{LDplot} (a). The decay is due to the fact that, under the pure selection equation,  the single fit mutant type ultimately goes to fixation, that is,
\begin{equation}\label{omega0}
\omega^{(0)}_\infty \defeq \lim_{t \to \infty} \omega^{(0)}_t = \delta_{x^{\text{mut}}},
\end{equation}
and the correlation vanishes for any point measure.
Let us now investigate how this behaviour changes in the presence of recombination. Motivated by the observation of \citet{stephansonglangley} and \citet{PfaffelhuberLehnertStephan} that the behaviour is crucially different depending on whether the selected site is to one side of or between $L$ and $R$, we distinguish between recombination events that separate $L$ and $R$ and those that do not. We therefore define $\rhosep$ as the total recombination rate between the sites $L$ and $R$, and $\rhonosep$ as  the remaining recombination rate, that is, the  rate at which $\{L,R\}$, as an intact entity, is separated from $i_\bullet$, see Fig.~\ref{SLR}:
\begin{equation}
\begin{split}\label{rhosepnosep}
\rhosep & \defeq \! \! \! \! \sum_{\substack{i \in S^\circ:  \\|S^\circ \cap C_i| = |S^\circ \cap D_i| = 1  }} \!\!\!\!\!\!\!\!\! \!\!\!\!\! \varrho_i^{} \quad \text{and}\\
\rhonosep \!&  \defeq \Big ( \sum_{i \in S^\circ}\varrho_i^{} \Big ) - \rhosep.
\end{split}
\end{equation}
We will then see that the dynamics of the LD \eqref{lddefinition} may be reparametrised in terms of $\rhosep$ and $\rhonosep$, thus eliminating the need to distinguish the different locations of $i_\bullet$. More precisely, the position of $i_\bullet$ affects the dynamics of the correlation only via $\rhosep$ and 
$\rhonosep$; in particular, $i_\bullet=2$ implies $\rhonosep=0$.

Let us first consider $\omega^{\nosep}_{}$, the solution (with initial condition $\omega_0^{}$) of Eq.~\eqref{recoselequation} with all $\varrho^{}_i$ for which $|S^\circ \cap C_i| = |S^\circ \cap D_i| = 1$ (and therefore $\rhosep$) set to 0 (note that, for $i_\bullet=2$, $\rhonosep = 0$ and therefore  $\omega^{\nosep}_{}=\omega^{(0)}$). We then have
\begin{lemma}\label{omegansasymptotics}
 For $\varrho_\textnormal{ns} =0$, one  has $\textnormal{Cor}(\omega_\infty^{\textnormal{ns}})=\textnormal{Cor}(\omega_\infty^{(0)})=0$. For $\varrho_\textnormal{ns} >0$,
\begin{equation*}
\begin{split}
\textnormal{Cor}(\omega_\infty^{\textnormal{ns}}) & = \int_0^\infty \varrho_\textnormal{ns} \ee^{-\varrho_\textnormal{ns} \tau} \omega_{\{L,R\},\tau}^{(0)} ((1,1)) \dd \tau \\
& \hphantom{=} - \int_0^\infty \varrho_\textnormal{ns} \ee^{-\varrho_\textnormal{ns} \tau} \omega_{\{L\},\tau}^{(0)} (1) \dd \tau \int_0^\infty \varrho_\textnormal{ns} \ee^{-\varrho_\textnormal{ns} \sigma} \omega_{\{R\},\sigma}^{(0)} (1) \dd \sigma.
\end{split}
\end{equation*} 
\end{lemma}
\begin{proof}
For $\varrho_\textnormal{ns} =0$, one has $\omega^{\nosep}_{}=\omega^{(0)}$ and the case is clear due to \eqref{omega0}.  For $\varrho_\textnormal{ns} > 0$ and $i_\bullet=1$, the claim 
follows immediately from Theorem~\ref{recursion} by letting $t \to \infty$ and  marginalisation; this carries over to  $i_\bullet=3$ via symmetry (or via \eqref{rec_general}). Alternatively, we can argue via the AIG: as $t$ tends to infinity, the age of the line ancestral to $\{L,R\}$, which is never split since $\rhosep = 0$, is exponentially distributed with mean $1 / \varrho_\textnormal{ns}$; hence the type of this line is sampled from $\omega_{\{L,R\}}^{(0)}$, evaluated at this exponential time.
\end{proof}

The numerical solution of $\big (\textnormal{Cor}(\omega_t^{\textnormal{ns}}) \big )_{t\geqslant 0}$  is shown in Fig.~\ref{LDplot} (a) for   various values of  $\rhonosep$. 
Let us start with the curve for $\rhonosep=0$, with its initial buildup of LD  followed by a decay to zero due to fixation of  the original `tail' $(x^{\text{mut}}_L,x^{\text{mut}}_R)=(1,1)$ together with $x_{i_\bullet}=0$; we take this as our `reference curve'. In contrast, for $\rhonosep>0$, the correlation is expected to remain positive in the long run, a phenomenon that can be explained from two different angles; see also the discussion surrounding Lemma~\ref{restrictionislinear} and Figure~\ref{aigpartition}.
\begin{enumerate}
\item
Forward in time, as the mutation at the selected site goes to fixation, the original tail associated with this mutant is replaced by a random sample 
$(x_L^{},x_R^{}) \in \{0,1\}^2$ from the population at the (random) time $\theta \sim \text{Exp}(\rhonosep)$.
\item
From a genealogical perspective, as we trace back the ancestral line of the neutral sites $L$ and $R$ (keep in mind that $L$ and $R$ remain glued together due to 
$\rhosep = 0$), this line is repeatedly (at rate $\rhonosep$) split from the ancestral line of $i_\bullet$, so that $L$ and $R$ only hitchike along with the orginal instance of $i_\bullet$ for the random time $\theta$ as above.
\end{enumerate}

In any case, $\textnormal{Cor}(\omega_t^{\textnormal{ns}})$ `decouples' from our reference curve $\textnormal{Cor}(\omega_t^{(0)})$  at roughly $\EE [\theta]$ (we will discuss the nature of this approximation below) and then remains constant since $\rhosep=0$; this happens the sooner the larger $\rhonosep$ since $\EE[\theta] = 1 / \rhonosep$. For  $\rhonosep = 0.7s$,  it happens almost instantly, so the correlation is nearly constant for all times; in fact, it is immediate from Lemma~\ref{omegansasymptotics} that $\textnormal{Cor}(\omega_\infty^{\textnormal{ns}})$ approaches $\textnormal{Cor}(\omega_0^{})$  as
$\varrho_\text{ns}^{} \to \infty$. In any case, the asymptotic distribution is not a point measure; and since the sample starts with positive LD  by assumption,  LD is preserved while $x_{i_\bullet}=0$ sweeps to fixation, whence we expect a non-zero limit of the correlation as $t \to \infty$. 

As our reference curve  $\textnormal{Cor}(\omega_t^{(0)})$ is not monotonic in time, the above `decoupling' mechanism implies that $\textnormal{Cor}(\omega_\infty^{\textnormal{ns}})$  is not monotonic in  $\varrho_{\textnormal{ns}}$. Roughly speaking, $\textnormal{Cor}(\omega_\infty^{\textnormal{ns}})$ will be maximal if the ancestral line of $\{L,R\}$ decouples from $i_\bullet^{}$ at the time of maximal LD in the reference curve. In our example, this will be around $t=1000$, which leads to the naive estimate
\begin{equation*}
\argmax_{\rhonosep} \omega_\infty^{\textnormal{ns}} \approx 10^{-3}.
\end{equation*}
A numerical analysis shows that the true value is $\rhonosep \approx 8.8 \cdot 10^{-4}$. The reason for the deviation  is twofold. The naive estimate implies two approximation steps:
\begin{equation}\label{approx}
\textnormal{Cor}(\omega_\infty^{\textnormal{ns}}) \approx \EE[\textnormal{Cor}(\omega_\theta^{(0)})] \approx \textnormal{Cor}(\omega_{\EE(\theta)}^{(0)}).
\end{equation}
The second approximation refers to the fact that $\textnormal{Cor}(\omega_t^{\textnormal{ns}})$ does not  `decouple' from our reference curve $\textnormal{Cor}(\omega_t^{(0)})$ precisely at $\EE[\theta]$, but at exponentially distributed  times, which makes a difference  due to the nonlinearity of $\textnormal{Cor}(\omega_\theta^{(0)})$ as a function of $\theta$.
However, the middle expression in \eqref{approx}  is still an approximation to the limiting value of the correlation; on the level of the formula  in Lemma~\ref{omegansasymptotics}, it amounts to replacing the double integral by the simple integral
\begin{equation*}
\int_0^\infty \rhonosep \ee^{-\rhonosep \tau} \omega_{\{L\},\tau}^{(0)} (1) \omega_{\{R\},\tau}^{(0)} (1) \dd \tau.
\end{equation*}
The deeper reason is that LD is a \emph{second}-order quantity involving two individuals, so that considering only the ancestry of a single individual can only yield a heuristic.

We now consider the effect of separating recombination.
\begin{prop}\label{separatingrecombination}
Let $\rhosep$ be as in \eqref{rhosepnosep}. 
Then,
\begin{equation*}
\textnormal{Cor}(\omega_t^{}) = \ee^{-\rhosep t} \textnormal{Cor} (\omega_t^{\nosep}).
\end{equation*}
\end{prop}
\begin{proof}
We argue via the AIG. The line ancestral to $\{L,R\}$ is hit at rate $\rhosep$ by a splitting event that separates the ancestral lines of $L$ and $R$. Thus, with probability $\ee^{-\rhosep t}$, no such splitting occurs on the ancestral line of $L$ and $R$, and thus the type agrees with the one  delivered by an AIG with $\rhosep=0$; its distribution is $\omega_t^{\nosep}$. On the other hand, if such a splitting has ocurred, the alleles at sites $L$ and $R$ are sampled independently, and are thus uncorrelated.
\end{proof}

The behaviour    is shown in Fig.~\ref{LDplot} (b) for the values of  $\rhonosep$ used in panel (a) --- in line with Proposition~\ref{separatingrecombination}, it is obtained by multiplying the functions in panel (a) with an  exponentially decaying factor. The resulting picture resembles Fig.~2 of~\citet{stephansonglangley}.  

\section{Outlook}
\label{sec:outlook}
Let us close by mentioning possible extensions as well as limitations of our approach. The properties of the selection term that enabled the recursive solution of the selection-recombination equation are satisfied more generally. Put informally, what is required is that this term only affects a single locus. Thus, extensions to more general selection, in particular frequency-dependent selection (of which  diploid selection under dominance is a special case), as well as mutation, can be treated in this way. We defer these treatments to forthcoming work.
In contrast, the methods presented here and in~\cite{Selrek} break down if one tries to incorporate multiple selected sites; new ideas must then be sought, see~\cite{canada} and the corresponding erratum. Also, the presented approach relies on the (conditional) independence of ancestral lines separated by recombination, which is destroyed by coalescence events that occur in the setting of finite populations or scalings other than the law of large numbers regime considered in this work.

\section*{Acknowledgements}
This work was funded by the German Research Foundation (Deutsche
Forschungsgemeinschaft, DFG) --- SFB 1283/2 2021 --- 317210226.

\bigskip

\bigskip
\end{document}